\documentclass[12pt]{article}
\usepackage{eurosym}
\usepackage[leqno]{amsmath}
\usepackage[amsmath,thmmarks,thref,framed]{ntheorem}
\usepackage{xcolor}
\usepackage[colorinlistoftodos]{todonotes}
\usepackage{authblk}
\usepackage{ascii}
\usepackage{amsfonts}
\usepackage{amssymb}
\usepackage{amstext}
\usepackage{bbm}
\usepackage{calrsfs}
\usepackage{color}
\usepackage{dsfont}
\usepackage{epsfig}
\usepackage{keystroke}
\usepackage{latexsym}
\usepackage{lmodern}
\usepackage{psfrag}
\usepackage{textcomp}
\usepackage{units}
\usepackage{url}
\usepackage{yfonts}
\usepackage{hyperref}
\usepackage{times}
\usepackage{anysize}

\hypersetup{colorlinks=true}
\hypersetup{urlcolor=black}
\hypersetup{linkcolor=blue}
\hypersetup{citecolor=blue}

\definecolor{black}{cmyk}{1,1,1,1}
\definecolor{blue}{cmyk}{1,1,0,0}

\theoremstyle{break}
\theoremheaderfont{\small\normalfont\sffamily}
\theorembodyfont{\itshape}
\theoremsymbol{\hfill\EOT}
\theoremseparator{:}
\newtheorem{theorem}{Theorem}
\theoremclass{Theorem}
\theoremstyle{plain}
\theoremheaderfont{\sc}
\theoremsymbol{\hfill\EOT}
\theoremseparator{.}
\newshadedtheorem{conjecture}{Conjecture}
\theoremstyle{break}
\theorembodyfont{\slshape}
\theoremsymbol{\hfill\EOT}
\newtheorem{lemma}{Lemma}
\theoremsymbol{\hfill\LF}
\theoremnumbering{arabic}

\theorembodyfont{\slshape}
\theoremsymbol{\hfill\LF}
\theoremnumbering{arabic}

\newshadedtheorem{corshaded}{Corollary}
\theoremstyle{plain}
\theoremheaderfont{\bfseries}
\theorembodyfont{\small\fontfamily{cmbr}\selectfont}
\theoremsymbol{\hfill\SI}
\theoremseparator{.}

\theoremstyle{break}
\theoremheaderfont{\bfseries}
\theorembodyfont{\small\fontfamily{cmbr}\selectfont}
\theoremsymbol{\hfill\SI}
\theoremseparator{.}

\theoremstyle{break}
\theoremheaderfont{\ttfamily}
\theorembodyfont{\small\sffamily}
\theoremsymbol{\ensuremath{\ast}}
\theoremseparator{.}

\theoremheaderfont{\itshape}
\theorembodyfont{\normalfont}
\theoremsymbol{\hfill\BS}
\theoremseparator{.}
\newtheorem{definition}{Definition}
\theoremstyle{break}
\theoremheaderfont{\normalfont\ttfamily}
\theorembodyfont{\small\fontfamily{cmbr}\selectfont}
\theoremsymbol{\hfill\ENQ}
\newtheorem{notation}{Notation}
\theoremheaderfont{\normalfont\ttfamily}
\theorembodyfont{\small\fontfamily{cmbr}\selectfont}
\theoremsymbol{\hfill\ENQ}

\theoremheaderfont{\normalfont\ttfamily}
\theorembodyfont{\small\fontfamily{cmbr}\selectfont}
\theoremsymbol{\hfill\ENQ}

\theoremheaderfont{\itshape}
\theorembodyfont{\upshape}
\theoremstyle{nonumberplain}
\theoremseparator{.}
\theoremsymbol{\End}
\newtheorem{proof}{Proof}
\def\1{\mathbf 1}
\def\a{\Gamma}

\def\B{\mathrm B}
\def\BL{\mathcal B}

\def\C{C^{*}}

\def\CAR{\mathrm{CAR}}
\def\CP{\mathbb{C}}

\def\eqt#1{\texorpdfstring{#1}{}}

\def\fp{\mathfrak{p}(\H,\a)}

\def\H{\mathcal H}
\def\h{\mathfrak h}

\def\I{\mathrm I}

\def\inner#1{\left< #1 \right>}

\def\N{\mathbb N}

\def\P{\mathrm{P}}

\def\sCAR{\mathrm{sCAR}(\H,\a)}

\def\states{\mathfrak E}

\def\tr{\mathrm{tr}}

\def\X{\mathcal X}

\def\Z{\mathbb Z}

\input{tcilatex}

\newcommand{\rmB}{\mathrm{B}}

\newcommand{\upd}{\up{d}}

\newcommand{\calF}{\mathcal{F}}
\newcommand{\frakF}{\mathfrak{F}}

\newcommand{\dwnfrakF}{\dwn{\frakF}}

\newcommand{\frakh}{\mathfrak{h}}

\newcommand{\rmH}{\mathrm{H}}

\newcommand{\dwnrmHS}{\dwn{\rmH\rmS}}

\newcommand{\rmI}{\mathrm{I}}

\newcommand{\dwni}{\dwn{i}}

\newcommand{\frakI}{\mathfrak{I}}

\newcommand{\rmJ}{\mathrm{J}}
\newcommand{\dwnj}{\dwn{j}}

\newcommand{\calK}{\mathcal{K}}
\newcommand{\dwnk}{\dwn{k}}

\newcommand{\rmK}{\mathrm{K}}

\newcommand{\rmL}{\mathrm{L}}
\newcommand{\dwnrmL}{\dwn{\rmL}}

\newcommand{\sfM}{\mathsf{M}}

\newcommand{\Mcom}[2]{\pint{\psi\dwn{#1}}{P\Gamma\psi\dwn{#2}}}

\newcommand{\upn}{^n}

\newcommand{\dwnn}{\dwn{n}}

\newcommand{\dwnomega}{\dwn{\omega}}

\newcommand{\omegaum}{\omega\dwn{P\dwnum}}
\newcommand{\omegadois}{\omega\dwn{P\dwndois}}
\newcommand{\dwnomegaPum}{\dwn{\omegaum}}
\newcommand{\dwnomegaPdois}{\dwn{\omegadois}}

\newcommand{\dwnP}{\dwn{P}}
\newcommand{\Pum}{P\dwnum}
\newcommand{\Pdois}{P\dwndois}
\newcommand{\dwnPum}{\dwn{\Pum}}
\newcommand{\dwnPdois}{\dwn{\Pdois}}

\newcommand{\rmR}{\mathrm{R}}

\newcommand{\dwnrmR}{\dwn{\rmR}}

\newcommand{\calS}{\mathcal{S}}

\newcommand{\rmS}{\mathrm{S}}

\newcommand{\rmT}{\mathrm{T}}

\newcommand{\rmU}{\mathrm{U}}

\newcommand{\rmUopluscom}[1]{\rmU_\oplus^{(n)}}

\newcommand{\til}[1]{\tilde{#1}}

\newcommand{\upmeio}{^{\frac{1}{2}}}

\newcommand{\conj}[1]{\left\{#1\right\}}

\newcommand{\norma}[1]{\Vert #1 \Vert}

\newcommand{\upperp}{^\perp}

\newcommand{\pint}[2]{\left< #1, #2 \right>}
\newcommand{\pintcom}[3]{\left< #1, #2 \right>\dwn{#3}}

\newcommand{\frakuno}{\mathfrak{1}}

\newcommand{\upast}{^\ast}

\newcommand{\updois}{^2}

\newcommand{\dwnum}{\dwn{1}}

\newcommand{\dwndois}{\dwn{2}}

\newcommand{\up}[1]{^{#1}}
\newcommand{\dwn}[1]{_{#1}}

\newcommand{\upmais}{\up{+}}
\newcommand{\upmenonos}{\up{-}}
\newcommand{\uppm}{\up{\pm}}

\newcommand{\noin}{\noindent}

\newcommand{\tilpsi}{\til{\psi}}

\marginsize{2cm}{2cm}{1cm}{1cm}

\begin{document}

\title{A $\Z_{2}$--Topological Index as a $\Z_{2}$--State Index}
\author{N. J. B. Aza \hspace{0.5cm} L. C. P. A. M. M\"ussnich \hspace{0.5cm} A. F. Reyes-Lega}
\maketitle

\begin{abstract}
\noindent Within the setting of infinite dimensional self--dual $\CAR$ $\C$--algebras describing fermions in the $\Z\upd$--lattice, we depart from the well--known Araki--Evans $\sigma(P\dwnum,P\dwndois)$ $\Z\dwndois$--index for quasi--free fermion states and rewrite it in terms of states, rather than in terms of basis projections. Furthermore, we reformulate results which relate equivalences of Fock representations with the index parity into results which relate equivalences of GNS representations and the associated index parity. \\

\bigskip

\end{abstract}


\tableofcontents

\section{Introduction}
This work is a preliminary result toward the stability of the \emph{Araki--Evans} $\mathbb{Z}_{2}$--Topological Index ($\Z\dwndois$--TI) for \emph{weakly interacting} fermions embedded on $\mathbb{Z}^{d}$--crystal lattices. The $\Z\dwndois$--TI was introduced in the quasi--free fermion setting, and can be used to study $\mathbb{Z}_{2}$ topological components and their physical meaning in this context. It is well--known that it discriminates parity sectors among quasi--free ground states \cite{EK98, Varilly}. Because of the strong connection between quasi--free ground states and \emph{basis projections} in the \emph{self--dual} framework, this $\Z\dwndois$--TI is written in terms of basis projections. Furthermore, the properties to which it is connected are stated in terms of the Fock representation for quasi--free states. For one recent work which discusses the use of the $\Z\dwndois$--TI in the context of lattice fermions, see, for instance, \cite{ARS22} and references therein.\par
Although we here present only a plain rewriting of the Araki--Evans $\Z\dwndois$--TI, this reformulation will allow for the further application of this index when sufficiently small (but not null) interactions are considered. Nonetheless, since no natural basis projections arise in the interacting case, this application first requires the $\Z_{2}$--TI to be written is terms of states, in such a way that its main properties are retained, but that no references to basis projections are needed. Also, following recent ideas which use the \emph{GNS representation} in the context of $\Z_{2}$--indexes \cite{ogata21cl,ogata2021class2}, we also reformulate in a simple way a result associated with the $\Z\dwndois$--TI into the GNS representation setting. Our main result, Theorem \ref{theorem:main}, is then a restatement of Araki's \cite[Theorems 6.14 and 6.15]{A87}.\par
On what regards further work, in \cite{AMR2}, it will be taken that the spectral gap associated with a family of \emph{free} fermion Hamiltonians in the same phase of matter will be uniform, in a way such that each of these has a unique gapped ground state. Then, we will consider properties of covariances of two--point correlation functions, as well as combine them with expansions of logarithms of generating functions associated with weakly interacting systems. Recall that time correlation functions appear in the perturbative expansion of (full) correlations for weakly interacting systems. Thus, under suitable assumptions we compare sets of quasi--free ground states of \emph{free} fermions systems with a set of states associated with fermions under weak interactions. This will allow us to approximate both sets via local perturbations \cite{nach21quasi}. We will then prove that the $\mathbb{Z}_{2}$--TI persists in the thermodynamic limit. Particularly we use the fermionic Renormalization Group equation, as well as efficient estimates associated with covariances of systems that are not necessarily translation invariant. For a full acount of our approach to weakly interacting fermions in the lattice, see \cite{LD1}.
\section{Mathematical Framework}
\subsection{Self--Dual CAR Algebra}

We refer to \cite{A87, LD1,ARS22} for an extensive treatment of self--dual $\CAR$ $\C$--algebras in the context of lattice fermions. We briefly introduce and recall some notations.

\begin{notation}
\begin{enumerate}
\item A norm on the generic vector space $\X$ is denoted by $\Vert \cdot \Vert _{\X}$ and the identity
  map of $\X$ by $\1_{\X}$. The space of all bounded linear operators on $(\X,\Vert \cdot \Vert _{\X}\mathcal{)}$
  is denoted by $\BL(\X)$. The unit element of any algebra $\X$ is always denoted by $\mathfrak{1}$,
  provided it exists. The scalar product of any Hilbert space $\X$ is denoted
  by $\langle \cdot,\cdot\rangle_{\X}$ and $\tr_{\X}$ represents the usual trace on $\BL(\X)$.

\item $(\H,\a)$ is a \emph{self--dual Hilbert space}, with $\H$ a complex and separable Hilbert space,
  with either even and finite or infinite dimension. $\conj{\psi\dwni \colon i\in\rmI}$ is an
  orthonormal basis for $\H$. $\a\colon\H\to\H$ is a \emph{complex conjugation} over $\H$.
  Also, we let $P = \a P\up{\perp} \a$ be a basis projection over $\H$, with range denoted by $\h\dwnP$.
  The set of all basis projections over $\H$ will be denoted by $\fp$. We will also denote by $\sCAR$
  the self--dual CAR algebra generated by a unit $\frakuno$ and by elements $\conj{\B(\varphi)}_{\varphi\in\H}$,
  indexed by $\H$ and which satisfy the usual self--dual CAR generating conditions%
  \footnote{Remember that these read, in the self--dual formalism, as $\varphi \mapsto \B\left(\varphi\right)^{*}$ being
  complex linear, $\B(\varphi)\upast$ being equal to $\B(\Gamma\varphi)$, and the canonical anticommutation
  relations being valid:

  \begin{equation}\label{CAR Grassmann III}
    \B(\varphi _{1})\B(\varphi_{2})^{*}+\B(\varphi_{2})^{* }\B(\varphi_{1})=
    \inner{\varphi _{1},\varphi_{2}}_{\H}\,\mathfrak{1}.
  \end{equation}
  }.
\end{enumerate} 
\end{notation}
Recall that, for any unitary operator $U\in\BL(\H)$ such that $U\a=\a U$, the family of
elements $\conj{\frakuno}\cup\conj{\B(U\varphi)}_{\varphi\in\H}$ satisfies the self--dual $\CAR$
generating conditions, and hence by the same token generates $\sCAR$. Such unitary operators
$U\in \BL(\H)$ are named \emph{Bogoliubov transformations}. To each Bogoliubov transformation $U$,
the unique $^{*}$--automorphism $\mathbf{\chi}_{U}$ that satisfies
\begin{equation}
  \mathbf{\chi}_{U}\left( \B(\varphi )\right) =\B(U\varphi ),\qquad \varphi \in \H,
  \label{eq:Bog_aut}
\end{equation}
is called the a \emph{Bogoliubov} $^{*}$\emph{--automorphism} associated with $U$.
$U=-\1_{\H}$, allows us to define \emph{odd} and \emph{even} elements of $\sCAR$:
\emph{odd} elements satisfy $\mathbf{\chi}_{-\1_{\H}}(A)=-A$, whereas \emph{even} elements satisfy
$\mathbf{\chi}_{-\1_{\H}}(A)=A$. Note that the subspace $\sCAR^{+}$ of even elements is a
sub--$\C$--algebra of $\sCAR$. We shall de note the subset of odd elements by $\sCAR\upmenonos$.


\subsection{States, Fock Representation and Projection Index}
Likewise in the context of $\sCAR$, we first introduce and recall some notations.
\begin{notation}
  $\omega$ is any state%
  \footnote{\label{foot:state} That is, a positive,
  normalized linear functional over $\sCAR$. In particular,
  $\overline{\omega\left(A\right)}=\omega\left(A^{*}\right)$,
  for any $A\in\sCAR$.
  }
  in $\sCAR\upast$, the set of linear functionals over the self--dual $\CAR$
  $\C$--algebra. We denote by $\states\subset\sCAR^{*}$ the set of all states over $\sCAR$.
  We let $(\H_{\omega},\pi_{\omega},\Omega_{\omega})$ denote the GNS triple associated with $\omega$. 
\end{notation}

\noin States $\omega\in\states$ are said to be \emph{quasi--free} when, if calculated
over odd monomials on $\B$ operators, yield zero, \emph{i.e.},
for all $N\in\N_{0}$ and $\varphi_{0},\ldots ,\varphi _{2N}\in \H$,
\begin{equation}
  \omega\left(\B\left(\varphi_{0}\right)\cdots\B\left(\varphi_{2N}\right)\right)=0, \label{ass O0-00}
\end{equation}
and, if calculated over even monomials, are given by a Pfaffian over two-point
correlations, \emph{i.e.}, for all $N\in\N$ and $\varphi_{1},\ldots,\varphi_{2N}\in\H$,
\begin{equation}
  \omega \left(\B\left(\varphi _{1}\right) \cdots \B\left(\varphi _{2N}\right)
  \right)=\mathrm{Pf}
  \left[\omega \left( \mathbb{O}_{k,l}\left( \B(\varphi _{k}),
    \B(\varphi_{l})\right)\right)\right]_{k,l=1}^{2N},\label{ass O0-00bis}
\end{equation}
where
\begin{equation}\label{eq:o_skew}
  \mathbb{O}_{k,l}\left(A_{1},A_{2}\right)\doteq
  \left\{
  \begin{array}{ccc}
    A_{1}A_{2} & \text{for} & k<l,\\
    -A_{2}A_{1} & \text{for} & k>l,\\
    0 & \text{for} & k=l.
  \end{array}\right.
\end{equation}
In Equation \eqref{ass O0-00bis}, $\mathrm{Pf}$ is the usual Pfaffian defined by
\begin{equation}
  \mathrm{Pf}\left[M_{k,l}\right]_{k,l=1}^{2N}\doteq
  \frac{1}{2^{N}N!}\sum_{\pi \in \mathcal{S}_{2N}}
  \left(-1\right)^{\pi}\prod\limits_{j=1}^{N}M_{\pi \left( 2j-1\right) ,\pi\left(2j\right)}\label{Pfaffian}
\end{equation}
for any $2N\times 2N$ skew--symmetric matrix $M\in \mathrm{Mat}\left(2N,\CP\right)$.
Note that \eqref{ass O0-00bis} is equivalent to the definition given either in
\cite[Definition 3.1]{A70} or in \cite[Equation (6.6.9)]{EK98}. Moreover, one can show that
a quasi--free state $\omega\in\states$ is described (uniquely) by a \emph{symbol}, that is,
a positive operator $S_{\omega}\in\BL(\H)$ such that
\begin{equation}
  0 \leq S_{\omega} \leq \1_{\H} \qquad\text{and}\qquad S_{\omega}+\a S_{\omega}\a=\1_{\H}
  \label{symbol}
\end{equation}
through the conditions
\begin{equation}
  \inner{\varphi _{1},S_{\omega}\varphi_{2}}_{\H}=\omega
  \left(\B(\varphi _{1})\B(\varphi_{2})^{*}\right),\qquad\varphi_{1},\varphi_{2}\in\H.\label{symbolbis}
\end{equation}
Conversely, any self--adjoint operator $P$ satisfying \eqref{symbol} uniquely defines a
quasi--free state through Equation \eqref{symbolbis}: 
\begin{equation}
  \omega_{P}\left(\B(\varphi _{1})\B(\varphi_{2})^{*}\right)=
  \inner{\varphi _{1},P\varphi_{2}}_{\H},\qquad\varphi_{1},\varphi_{2}\in\H.\label{symbolbis2}
\end{equation}
Thus, any basis projection associated with $(\H,\a)$ can be seen as a symbol for
a quasi--free state over $\sCAR$. Quasi--free states defined by basis projectios are
\emph{pure} and it will be refered to as \emph{Fock states} \cite[Lemma 4.3]{A70}.\\
An example of a quasi--free state is provided by the tracial state:
\begin{definition}[Tracial state]\label{def:trace state}
  The tracial state $\tr\in\states$ is the quasi--free state with
  symbol $S_{\tr}\doteq\frac{1}{2}\mathbf{1}_{\H}$.
\end{definition}
Take $P\in\fp$, with range $\h_{P}$. For any $n\in\N$ and $\varphi_{1},\ldots,\varphi_{n}\in\h_{P}$,
we denote by $\varphi_{1}\wedge\cdots\wedge \varphi_{n}$
the \emph{completely antisymmetric $n$--linear form associated with}%
\footnote{
One can implement it, for instance, through
\begin{align*}
  \varphi_{1}\wedge\cdots\wedge \varphi_{n}(\phi\upast\dwnum,\dots,\phi\upast\dwnn)
  = \det(\phi\upast\dwni(\varphi\dwnj))\dwn{i,j=1}\upn
  ,
\end{align*}
for all $\phi\upast\dwnum,\dots,\phi\dwnn\upast$ elements of the dual $\H\upast$.
}
$\varphi_{1},\ldots,\varphi_{n}$.
Recall that 
\begin{align*}
  \varphi_{1}\wedge\cdots\wedge
  \varphi_{n}=
  \varepsilon_{\pi}\varphi_{\pi(1)}\wedge\cdots\wedge
  \varphi_{\pi(n)},\qquad \varphi_{1},\ldots,\varphi_{n}\in\h_{P},  
\end{align*}
with $\varepsilon_{\pi}$ equal to $+1$ or $-1$ if the permutation $\pi\in\calS\dwnn$
is even or odd, respectively. For $n\in\N_{0}$, we define $\wedge^{0}\h_{P}\doteq\CP$, and,
for $n\in\N$ we define
\begin{align*}
  \wedge^{n}\h_{P}\doteq\mathrm{lin}\{\varphi_{1}\wedge\cdots\wedge
  \varphi_{n}\colon \varphi_{1},\ldots,\varphi_{n}\in\h_{P}\}
  .
\end{align*}
We introduce an inner product in $\wedge\upn\h\dwnP$ through
\begin{align*}
  \pintcom{\varphi_{1}\wedge\cdots\wedge \varphi_{n}}{\phi_{1}\wedge\cdots\wedge\phi_{n}}{\wedge\upn\h\dwnP}
  \doteq
  \det
  \left(\pintcom{\varphi\dwni}{\phi\dwnj}{\H}
  \right)\dwn{i,j=1}\upn
  .
\end{align*}
Therefore, if we let
\begin{align}\label{eq:ferm_fock}
  \wedge\h_{P}\doteq\bigoplus_{n\geq0}\wedge^{n}\h_{P}
  ,
\end{align}
we might naturally endow it with a pre--Hilbert structure, the completion of which we call the \emph{fermionic Fock space associated with $\h\dwnP$}, denoted by $\calF(\h\dwnP)$. We recall that ``$\wedge$'' can be properly defined as a product over $\wedge\h_{P}$, making it an associative algebra with unit. This product extends over to the fermionic Fock space.

We denote the \emph{vacuum vector} by $\Omega\in\wedge\h_{P}$, and it is such that
$[\Omega]_{0}\doteq1\in\h_{P}^{0}$ and $[\Omega]_{n}\doteq0\in\h_{P}^{n}$ for $n\geq1$.
Note that the vacuum vector is the unit of the algebra $\wedge\h_{P}$. The maps
$a^{*}\colon\wedge\h_{P}\to\wedge\h_{P}$ and $a\colon\wedge\h_{P}\to\wedge\h_{P}$
defined by, for all $\xi,\,\zeta$ in $\wedge\h\dwnP$,
\begin{align*}
  a\upast(\xi)\zeta & = \xi\wedge\zeta
  , \quad \\
  \forall\eta\in\wedge\h\dwnP,\,
  \pintcom{a(\xi)\zeta}{\eta}{\wedge\h\dwnP} & =
  \pintcom{\zeta}{\xi\wedge\eta}{\wedge\h\dwnP}
  ,
\end{align*}
and extended over to the Fock space,
are the so--called \emph{creation} and \emph{annihilation} 
operators, respectively. They are shown to be bounded operators.
Among other properties, they satisfy $a(\varphi)\Omega=0$,
and $a^{*}(\varphi)\Omega=\varphi$, for all $\varphi\in\h_{P}$. Here, for $\varphi\in\h_{P}$,
the involution of $a(\varphi)\in\BL(\wedge\h_{P})$, namely, $a(\varphi)^{*}\in\BL(\wedge\h_{P})$,
is canonically identified with $a^{*}(\varphi)$, i. e., $a^{*}(\varphi)\equiv a(\varphi)^{*}$
\cite{attal2006open}. Additionally, the $\CAR$ hold:
\begin{align*}
  a(\varphi_{1})a^{*}(\varphi_{2})+a^{*}(\varphi_{2})a(\varphi_{1})=
  \inner{\varphi_{1},\varphi_{2}}_{\h_{P}}\1_{\wedge\h_{P}},\quad a(\varphi_{1})a(\varphi_{2})+a(\varphi_{2})a(\varphi_{1})=0.  
\end{align*}
Hence, the family of operators $\{a(\varphi)\}_{\varphi\in\h_{P}}$ and $\1_{\wedge\h_{P}}$
generate a $\CAR$ $\C$--algebra. By \cite[Lemma 3.3]{A68},
there is a $^{*}$--isomorphism from the self--dual $\CAR$ algebra onto the $\CAR$
generated by the creation and annihilation operators acting on the Fock space.  
This allows us to implement the so-called \emph{(fermionic) Fock representation} of the $\sCAR$ algebra.
Explicitly, for any $P\in\fp$ and $\varphi\in\H$, we write
\begin{equation}\label{eq:fock_rep}
  \pi_{P}(\B(\varphi))\doteq a(P\varphi)+a^{*}(\a P^{\perp}\varphi),\qquad \varphi\in\H
  ,
\end{equation}
and require that $\pi\dwnP$ is extended by linearity, in that it preserves products and adjoints.
When considering the representation $(\wedge\h\dwnP, \pi\dwnP)$,
it can be shown, then, that the above \emph{Fock states} $\omega\dwnP$ can be written as
\begin{equation}\label{eq:focK_state}
  \omega_{P}\left(A\right) =
  \inner{\Omega,\pi_{P}(A)\Omega}_{\wedge\h_{P}},\qquad A\in\sCAR
  .
\end{equation}
Before defining the \emph{projection index} we introduce one further notation.
Take the even and odd parts $\sCAR^{\pm}\subset\sCAR$ (see expression \eqref{eq:Bog_aut}).
Let $\pi_{P}$ be the fermionic Fock representation associated
to $P$ given by \eqref{eq:fock_rep}. As shown by \cite{A87}, $\pi_{P}$ can be decomposed
as two disjoint irreducible representations : $\pi_{P}=\pi_{P}^{+}\oplus\pi_{P}^{-}$.
These are defined in the following way: let 
\begin{align*}
  \H_{P}^{\pm}\doteq \overline{\pi\dwnP(\sCAR\uppm)\Omega\dwnP}.
\end{align*}
Then, $\pi\dwnP\upmais$ is the representation obtained by the restriction of
$\pi\dwnP(\sCAR\upmais)$ to $\H\dwnP\upmais$, and $\pi\dwnP\upmenonos$ is the
representation obtained by the restriction of $\pi\dwnP(\sCAR\upmais)$ to $\H\dwnP\upmenonos$.

We then define the \emph{$\Z_{2}$--projection index} ($\Z_{2}$--PI) as follows:
$\Z_{2}$--PI is the function $\sigma\colon \fp\times\fp\to\Z_{2}$ defined by, 
for each $P_{1}$, $P_{2}\in\fp$, 
\begin{equation}\label{eq:top_index}
  \sigma(P_{1},P_{2})\doteq(-1)^{\dim(P_{1}\land P_{2}^{\perp})}
  ,
\end{equation}
where $P\dwnum\land P\dwndois\upperp$ denotes $P\dwnum\H\cap P\dwndois\upperp\H$.
It is a fact that $\sigma(P_{1},P_{2})$ gives an equivalence criterion for the
quasi--free states $\omega_{P_{1}}$ and $\omega_{P_{2}}$, when restricted to the even part
$\sCAR^{+}$ of the self--dual $\C$--algebra $\sCAR$. More generally,
the Shale--Stinespring Theorem states that two Fock representations $\pi_{P_{1}}$
and $\pi_{P_{2}}$ associated with $P_{1},P_{2}\in\fp$ are unitarily equivalent if,
and only if, $P_{1}-P_{2}\in\BL(\H)$ is a \emph{Hilbert--Schmidt class} operator \cite{Varilly}.
Moreover, the representations $\pi\dwn{P\dwnum}^{\pm}$
and $\pi\dwn{P\dwndois}^{\pm}$ are (unitarily) equivalent if, and only if,
$P_{1}-P_{2}\in\BL(\H)$ is a Hilbert--Schmidt operator \emph{and} $\sigma(P_{1},P_{2})$
equals $+1$. On the other hand, the representations $\pi\dwn{P\dwnum}^{\pm}$ and
$\pi\dwn{P\dwndois}^{\mp}$ are (unitarily) equivalent if, and only if, $P_{1}-P_{2}$ is a
Hilbert--Schmidt operator \emph{and} $\sigma(P_{1},P_{2})$ is $-1$. See \cite[Theorem 6.15]{A87}. 
\subsection{A Certain Space ``\eqt{$\frakF$}'' for Functionals}
For our purposes, the following constructions are convenient\footnote{We note that the notation ``$\frakF$'' carries no further meaning, apart from being an abbreviating symbol.}. Firstly, we introduce for, $f,g\in\sCAR^{*}$, a sesquilinear form on $\sCAR^{*}$:
\begin{equation}
  \label{eq:cond1}
  \inner{f,g}_{\sCAR^{*}}\doteq
  \sum_{i,j\in\I}
  \overline{f\left(\B(\psi_i)\B(\psi_j)\upast\right)}
  g\left(\B(\psi_{i})\B(\psi_{j})\upast\right)
  ,
\end{equation}
where $\left\{\psi_{i}\colon i\in\I\right\}$ is an orthogonal basis of $\H$.
For it to be meaningful, it should be restricted to elements
$f$ and $g$ for which $\inner{f,g}_{\sCAR^{*}}< \infty$.
Note that, if $\H$ has infinite dimension, and $P\in\fp$, the quasi--free state $\omega_{P}\in\states$ \emph{not} satisfies $\pint{\omega\dwnP}{\omega\dwnP}_{\sCAR^{*}}<\infty$. See \eqref{symbolbis2}.
It will be useful, nonetheless, to introduce
the following quantity, for any $\omega_{1},\omega_{2}\in\states$:
\begin{align}\label{eq:cond2}
\mathcal{N}(\omega_{1},\omega_{2})&\doteq
\inner{\omega_{1},2\tr-\omega_{2}}_{\sCAR^{*}},
\end{align}
where $\tr\in\states$ is the tracial state of Definition \ref{def:trace state}.
It has the following properties:
\begin{lemma}\label{lemma:number}
  Let $\omega_{1},\omega_{2}\in\states$ be states over $\sCAR$.
  Then, the quantity $\mathcal{N}(\omega_{1},\omega_{2})$, given by \eqref{eq:cond2},
  is basis independent and $\mathcal{N}(\omega_{1},\omega_{2})=\mathcal{N}(\omega_{2},\omega_{1})$.
  Moreover, if $P_{1},P_{2}\in\fp$ are basis projections such that $P_{1}-P_{2}\in\BL(\H)$
  is a Hilbert--Schmidt operator, and $\omega_{P_{1}},\omega_{P_{2}}\in\states$ are their
  associated quasi--free states, we have 
  \begin{align}\label{eq:number}
    \mathcal{N}(\omega_{P_{1}},\omega_{P_{2}})=\dim(P\dwnum\land P\dwndois\upperp)\in\N_{0}. 
  \end{align}
\end{lemma}
\begin{proof}
Let $\left\{\psi_{i}\colon i\in\I\right\}$ be an orthogonal basis of $\H$. Basis--independence for $\mathcal{N}(\omega_{1},\omega_{2})$ derives from the fact that the summand is antilinear in $\psi\dwni$ for $\omega_{1}$ and linear
  in $\psi\dwni$ for $\omega_{2}$, and likewise for $\psi\dwnj$. In order
  to prove the symmetry, \emph{i.e.}, $\mathcal{N}\left(\omega_{1},\omega_{2}\right)=\mathcal{N}\left(\omega_{2},\omega_{1}\right)$,
  note that one can choose for $\H$ the basis
  \begin{align}\label{eq:basis_1}
    \conj{\psi\dwnj\colon j\in\rmJ} \cup \conj{\Gamma\psi\dwnj\colon j\in\rmJ},
  \end{align}
  which splits it into the direct sum $\frakh\dwn{P\dwnum}\oplus\frakh\dwn{P\upperp\dwnum}$,
  with $\conj{\psi\dwnj\colon j\in\rmJ}$ being the basis for $\frakh\dwn{P\dwnum}$.
  Then note that, by \eqref{CAR Grassmann III} and footnote \ref{foot:state},
  we are able to write $\mathcal{N}(\omega_{1},\omega_{2})$ as
  \begin{align*}
    \mathcal{N}(\omega_{1},\omega_{2}) &=
    \sum_{i,j\in\I}
    \overline{\omega_{1}\left(\B(\psi_i)\B(\psi_j)\upast\right)}
    \left(\delta_{i,\,j}-\omega_{2}\left(\B(\psi_{i})\B(\psi_{j})\upast\right)\right)\\
    & =\sum\dwn{i,j\in\rmI}\omega_{1}
    \left(\rmB(\psi\dwnj)\rmB(\psi\dwni)\upast\right)\omega_{2}
    \left(\rmB(\psi\dwnj)\upast\rmB(\psi\dwni)\right)\\
    & = \sum\dwn{i,j\in\rmI}\omega_{1}\left(\rmB(\Gamma\psi\dwnj)\upast\rmB(\Gamma\psi\dwni)\right)\omega_{2}
    \left(\rmB(\Gamma\psi_{j})\rmB(\Gamma\psi_{i})\upast\right) \\
    & = \sum\dwn{i,j\in\rmI}\omegaum\left(\rmB\left(\tilpsi\dwnj\right)\upast
    \rmB\left(\tilpsi\dwni\right)\right)\omega_{2} \left(\rmB\left(\tilpsi\dwnj\right)
    \rmB\left(\tilpsi\dwni\right)\upast\right)\\
    & =\sum\dwn{i,j\in\rmI}\omega_{2} \left(\rmB\left(\tilpsi\dwnj\right)
    \rmB(\tilpsi\dwni)\upast\right)\omega_{1}\left(\rmB\left(\tilpsi\dwnj\right)\upast
    \rmB\left(\tilpsi\dwni\right)\right) \\
    & = \mathcal{N}(\omega_{2},\omega_{1}),
  \end{align*}
  where $\conj{\tilpsi\dwni}$ is just the previous basis, only reordered as
  $\conj{\Gamma\psi\dwnj}\cup\conj{\psi\dwnj}$.\\
  Finally, if $P_{1}-P_{2}\in\BL(\H)$ is a Hilbert--Schmidt operator, then $P\dwnum\land P\dwndois\upperp$
  is a vector subspace of $\H$ with finite dimension \cite[Page 95]{A87}.
  Hence, in order to perform a calculation for $\mathcal{N}(\omega_{P_{1}},\omega_{P_{2}})$,
  we may choose, as orthonormal basis for $\H$, the following basis
  \begin{align*}
    \conj{\psi\dwni \colon i\in\rmI} =
    \conj{\tilde{\psi}\dwnk \colon k\in\rmK}
    \cup
    \conj{\breve{\psi}\dwn{k'}\colon k'\in\rmK'},
  \end{align*}
  where $\conj{\tilde{\psi}_{k}\colon k\in\mathrm{K}}$ is an orthonormal basis
  for $P_{1}\H\cap P_{2}^{\perp}\H$, while $\conj{\breve{\psi}\dwn{k'}\colon k'\in\mathrm{K}'}$
  is an orthonormal basis for its orthogonal complement, with $\mathrm{K},\mathrm{K}'$
  appropriate index sets. In this case, from expressions
  \eqref{CAR Grassmann III}, \eqref{symbolbis2} and footnote \ref{foot:state},
  one obtains that expression \eqref{eq:cond2} can be written as
  \begin{align*}
    \mathcal{N}(\omega_{P_{1}},\omega_{P_{2}}) & =
    \sum\dwn{i,j\in\rmI}\omega_{P_{1}}\left(\rmB(\psi_{j})\rmB(\psi_{i})\upast\right)
    \omegadois\left(\rmB(\psi_{j})\upast\rmB(\psi_{i})\right)\\
    & = \sum\dwn{i,j\in\rmI}\inner{\psi_{j},P_{1}\psi_{i}}_{\H}\inner{\psi_{i},
      \left(\1_{\H}-P_{2}\right)\psi_{j}}_{\H}\\
    & = \sum\dwn{i,j\in\rmK}\inner{\tilde{\psi}_{j},P_{1}\tilde{\psi}_{i}}_{\H}
    \inner{\tilde{\psi}_{i},\left(\1_{\H}-P_{2}\right)\tilde{\psi}_{j}}_{\H}\\
    & = \sum\dwn{i\in\rmK}\inner{\tilde{\psi}_{i},\tilde{\psi}_{j}}_{\H}\\
    & = \dim(P\dwnum\land P\dwndois\upperp).
  \end{align*}
\end{proof}
Observe that in order to avoid degeneracy of the sesquilinear form
$\inner{\,\cdot\,,\,\cdot\,}_{\sCAR^{*}}$, we can apply a standard procedure,
like that performed in \cite[Page 532]{lang12}, and turn it non--degenerate. For any $f\in\sCAR^{*}$
we define left and right kernels, respectively, by
\begin{align*}
  \calK\dwnrmL \doteq
  \conj{f\colon \forall g,\inner{f,g}_{\sCAR^{*}} = 0},\quad
  \calK\dwnrmR \doteq
  \conj{f\colon \forall g, \inner{g,f}_{\sCAR^{*}} = 0},
\end{align*}
which, by the construction of $\pint{\,\cdot\,}{\,\cdot\,}_{\sCAR^{*}}$,
are equal -- whence we call them both $\calK$. The set $\calK$ is a
subspace of $\sCAR\upast$, and we can define, for each $f\in\sCAR^{*}$, cosets of $\calK$:
\begin{align*}
  [f] \doteq f + \calK.
\end{align*}
We may then create the union of all cosets thus defined,
\begin{align*}
\frakF \doteq\bigcup\dwn{f\in\sCAR\upast}\conj{[f]},
\end{align*}
and introduce in such a union a vector space structure. The sesquilinear form
of equation \eqref{eq:cond1} yields, then, a non--degenerate form in $\frakF$,
unambiguously given by
\begin{align*}
  \pintcom{[f]}{[g]}{\frakF} \doteq\pint{f}{g}_{\sCAR^{*}}, \quad\text{for any}\quad f,g\in\sCAR^{*},
\end{align*}
and, hence, $\frakF$ is a pre--Hilbert space, the completion of which -- through
the norm induced by $\pintcom{\,\cdot\,}{\,\cdot\,}{\frakF}$, that is,
\begin{align*}
  \norma{[f]}\dwn{\frakF} \doteq
  \left(
  \sum\dwn{i,j\in\rmI}
  \left|
  f\left(\B(\psi_i)\B(\psi_j)\upast\right)
  \right|\updois
  \right)\upmeio
\end{align*}
-- will be denoted by $\hat{\frakF}$. We shall write, without confusion, $\inner{f,g}_{\frakF}$ and $\norma{f}\dwn{\frakF}$.
\section{Main results}
\begin{lemma}[$\Z\dwndois$--PI as a $\Z_{2}$--State Index]\label{lemma:states}
  Let $P_1,P_2\in\fp$, and let $\omega_{P_{1}},\omega_{P_{2}}\in\states$ be their
  respective quasi--free states. Then:
  \begin{enumerate}
  \item\label{enu1:proof1} $P_1-P_2\in\BL(\H)$ is a Hilbert--Schmidt operator if, and only if,
    \begin{align*}
      \norma{\omega\dwn{P\dwnum}-\omega\dwn{P\dwndois}}\dwn{\frakF} < \infty.
    \end{align*}
  \item\label{enu3:proof1} The $\Z_{2}$--PI given by \eqref{eq:top_index}
    can be rewritten in terms of the quantity \eqref{eq:number} as
    \begin{align*}
      \sigma(P_{1},P_{2}) = (-1)\up{\mathcal{N}(\omega_{P_{1}},\omega_{P_{2}})} \doteq
     \sigma(\omega_{P_{1}},\omega_{P_{2}}).
   \end{align*}
 \end{enumerate}
\end{lemma}
\begin{proof}
  We begin by proving \ref{enu1:proof1}. It is a simple result.
  First, recall that, for $\rmT\in\BL(\H)$, the Hilbert--Schmidt norm is given by
\begin{align*}
    \norma{\rmT}\dwnrmHS \doteq
    \left(
    \sum\dwn{i\in\rmI}
    \norma{\rmT\psi\dwni}\dwn{\H}\updois
    \right)^{\frac{1}{2}}.
  \end{align*}
Explicitly, for $P_{1}-P_{2}\in\BL(\H)$ we are able to write
\begin{align*}
    \norma{P\dwnum-P\dwndois}\dwnrmHS\updois & =
    \sum\dwn{i\in\rmI}
    \norma{(P\dwnum-P\dwndois)\psi\dwni}\dwn{\H}\updois\\
    & = \sum\dwn{i\in\rmI}\left(\inner{\psi_{i},P_{1}\psi_{i}}_{\H}
    +\inner{\psi_{i},P_{2}\psi_{i}}_{\H} -\inner{\psi_{i},P_{1}P_{2}\psi_{i}}_{\H} -
    \inner{\psi_{i},P_{2}P_{1}\psi_{i}}_{\H} \right).
  \end{align*}
Suppose, then, that $P\dwnum-P\dwndois\in\BL(\H)$ is Hilbert--Schmidt,
i.e., $\norma{P\dwnum-P\dwndois}\dwnrmHS\updois<\infty$.
Notice that under some calculations one can rewrite $\norma{P\dwnum-P\dwndois}\dwnrmHS\updois$ as
\begin{align*}
  \norma{P\dwnum-P\dwndois}\dwnrmHS\updois =\sum\dwn{i\in\rmI}
  \left(\pint{(\1_{\H}-P\dwndois)\psi\dwni}{P\dwnum\psi\dwni}_{\H}
  + \pint{(\1_{\H}-P\dwnum)\psi\dwni}{P\dwndois\psi\dwni}_{\H}\right).
\end{align*}
Since the Hilbert--Schmidt norm is basis independent,
one can choose for $\H$ the basis as in Expression \eqref{eq:basis_1},
which splits it into the direct sum $\frakh\dwn{P\dwnum}\oplus\frakh\dwn{P\upperp\dwnum}$,
with $\conj{\tilde{\psi}\dwnj\colon j\in\rmJ}$ being the basis for $\frakh\dwn{P\dwnum}$.
It follows that
\begin{align*}
\norma{P\dwnum-P\dwndois}\dwnrmHS\updois&=\sum\dwn{j\in\rmJ}
\pint{(\1_{\H}-P\dwndois)\tilde{\psi}\dwnj}{\tilde{\psi}\dwnj}_{\H}+\sum\dwn{j\in\rmJ}
\pint{\Gamma\tilde{\psi}\dwnj}{P\dwndois\Gamma\tilde{\psi}\dwnj}_{\H}\\ 
&= 2 \sum\dwn{j\in\rmJ}\pint{\tilde{\psi}\dwnj}{(\1_{\H}-P\dwndois)\tilde{\psi}\dwnj}_{\H},
\end{align*}
and one hence concludes that $P\dwnum-P\dwndois\in\BL(\H)$ is Hilbert--Schmidt if, and only if
\begin{align*}
\sum\dwn{j\in\rmJ}\pint{\tilde{\psi}\dwnj}{(\1_{\H}-P\dwndois)\tilde{\psi}\dwnj}_{\H}<\infty.
\end{align*}
Now, consider the quantity
\begin{align*}
  \norma{\omega\dwn{P\dwnum}-\omega\dwn{P\dwndois}}\dwn{\frakF}\updois =
  \sum\dwn{i,j\in\rmI}&
  \overline{\left(\omega\dwn{P\dwnum}(\B(\psi\dwni)\B(\psi\dwnj)\upast)
    -\omega\dwn{P\dwndois}(\B(\psi\dwni)\B(\psi\dwnj)\upast)\right)} \\
  &\hspace{3cm}\left(\omega\dwn{P\dwnum}(\B(\psi\dwni)\B(\psi\dwnj)\upast)
  -\omega\dwn{P\dwndois}(\B(\psi\dwni)\B(\psi\dwnj)\upast)\right),
\end{align*}
which, in face of Expression \eqref{symbolbis2}, is equivalent to
\begin{align*}
\norma{\omega\dwn{P\dwnum}-\omega\dwn{P\dwndois}}\dwn{\frakF}\updois & =
\sum\dwn{i,j\in\rmI}\left(\overline{\pint{\psi\dwni}{P\dwnum\psi\dwnj}_{\H}} - \overline{\pint{\psi\dwni}{P\dwndois\psi\dwnj}_{\H}}\right)
\left(\pint{\psi\dwni}{P\dwnum\psi\dwnj}_{\H} - \pint{\psi\dwni}{P\dwndois\psi\dwnj}_{\H}\right)\\
& = \sum\dwn{i,j\in\rmI}\left(\left|\pint{\psi\dwni}{P\dwnum\psi\dwnj}_{\H}\right|\updois + \left|\pint{\psi\dwni}{P\dwndois\psi\dwnj}_{\H}\right|\updois\right.\\
&\left.\hspace{2cm} -\pint{\psi\dwni}{P\dwndois\psi\dwnj}_{\H}\pint{\psi\dwnj}{P\dwnum\psi\dwni}_{\H}-\pint{\psi\dwni}{P\dwnum\psi\dwnj}_{\H}\pint{\psi\dwnj}{P\dwndois\psi\dwni}_{\H}\right).
\end{align*}
This quantity is basis independent, with the basis choices
for the sums over $i$ and $j$ not being necessarily equal.
Then, let, for the sum over $j$, the basis be given by Equation \eqref{eq:basis_1}.
Similarly, let the basis for sum over $i$ be that which splits $\H$ into the direct sum
$\frakh\dwn{P\dwndois}\oplus\frakh\dwn{P\upperp\dwndois}$ be denoted by
\begin{align*}
\conj{\breve{\psi}\dwnj \colon j\in\rmJ} \cup \conj{\Gamma\breve{\psi}\dwnj \colon j\in\rmJ},
\end{align*}
with $\conj{\breve{\psi}\dwnj \colon j\in\rmJ}$ being the basis for $\frakh\dwn{P\dwndois}$,
c.f., \eqref{eq:basis_1}. We note that, since $\frakh\dwn{P\dwnum}$ and $\frakh\dwn{P\dwndois}$
have the same dimension, we can choose the same index set $\rmJ$ for both their basis.
Straightforward calculations then show that
\begin{align*}
\norma{\omega\dwn{P\dwnum}-\omega\dwn{P\dwndois}}\dwn{\frakF}\updois =2\sum\dwn{i,j\in\rmJ}\left|\pint{\breve{\psi}\dwni}{\Gamma\tilde{\psi}\dwnj}_{\H}\right|\updois.
\end{align*}
Therefore, $\norma{\omega\dwn{P\dwnum}-\omega\dwn{P\dwndois}}\dwn{\frakF}\updois$ is finite if, and only if,
\begin{align*}
\sum\dwn{i,j\in\rmJ}\left|\pint{\breve{\psi}\dwni}{\Gamma\tilde{\psi}\dwnj}_{\H}\right|\updois <\infty.
\end{align*}
We nonetheless observe that the projection of $\Gamma\tilde{\psi}\dwnj$ onto $\frakh\dwn{P\dwndois}$ is given by
\begin{align*}
P\dwndois\Gamma\tilde{\psi}\dwnj =\sum\dwn{i\in\rmJ}\pint{\breve{\psi}\dwni}{\Gamma\tilde{\psi}\dwnj}_{\H}\breve{\psi}\dwni,
  \end{align*}
  so that
  \begin{align*}
    \pint{P\dwndois\Gamma\tilde{\psi}\dwnj}{P\dwndois\Gamma\tilde{\psi}\dwnj}_{\H} =
    \sum\dwn{i\in\rmJ}
    \left|\pint{\breve{\psi}\dwni}{\Gamma\tilde{\psi}\dwnj}_{\H}
    \right|\updois.
  \end{align*}
  But, since
  $\pint{P\dwndois\Gamma\tilde{\psi}\dwnj}{P\dwndois\Gamma\tilde{\psi}\dwnj}_{\H} =
  \pint{\tilde{\psi}\dwnj}{(\1_{\H}-P\dwndois)\tilde{\psi}\dwnj}_{\H}$,
  we have
  \begin{align*}
    \sum\dwn{j\in\rmJ}
    \pint{\tilde{\psi}\dwnj}{(\1_{\H}-P\dwndois)\tilde{\psi}\dwnj}_{\H} =
    \sum\dwn{i,j\in\rmJ}
    \left|
    \pint{\breve{\psi}\dwni}{\Gamma\tilde{\psi}\dwnj}_{\H}
    \right|\updois,
  \end{align*}
  in which case statement \ref{enu1:proof1} is proven.
  Part \ref{enu3:proof1} is clear from Lemma \ref{lemma:number}--quantity \eqref{eq:number} and Definition of the $\Z_{2}$--PI,
  given by \eqref{eq:top_index}. 
\end{proof}
In order to state our main Theorem, as well as its proof, some considerations are in order.
Firstly, for a basis projection $P\in\fp$, we explicitly construct the GNS representation
associated with its quasi--free state $\omega\dwnP$. This is a well--known result and we write
it for completeness. See, for instance, \cite[Chap. 6]{EK98}. As was discussed
for equation \eqref{symbolbis2} and comments around it, $\omega\dwnP$ is completely defined by
two--point correlations, satisfying \eqref{ass O0-00}--\eqref{Pfaffian}.
If we choose for $\H$ the basis given by \eqref{eq:basis_1}, namely,
\begin{align*}
  \conj{\tilde{\psi}\dwnj\colon j\in\rmJ} \cup \conj{\Gamma\tilde{\psi}\dwnj\colon j\in\rmJ},  
\end{align*}
with $P$ instead of $P\dwnum$, this yields a useful form for computations of the type $\omega_{P}(A)$,
for $A\in\sCAR$. Let $A\equiv\B(\psi_{1})\cdots\B(\psi_{2N})$, where, for $i\in\{1,\ldots, 2N\}$, $\psi_{i}$ is a
basis element of $\H$, with $N\in\N$. Observe that $\omega\dwnP(A)$ is given by
\begin{align*}
  \omega_{P}\left(\B\left(\psi_{1}\right) \cdots \B\left(\psi_{2N}\right) \right)=
  \mathrm{Pf}\left[\omega_{P}\left( \mathbb{O}_{k,l}\left( \B(\psi_{k}),\B(\psi_{l})\right)\right)\right]_{k,l=1}^{2N},  
\end{align*}
where, for $k,l\in\{1,\ldots,2N\}$, $\mathbb{O}_{k,l}$ is defined by \eqref{eq:o_skew}. Note that for $P\in\fp$, the $2N\times 2N$ matrix 
\begin{equation*}
\sfM_{k,l}^{2N}\doteq\left[\mathbb{O}_{k,l}(\inner{\psi_{k},P\Gamma\psi_{l}}_{\H})\right] _{k,l=1}^{2N}
\end{equation*}
is skew--symmetric and satisfies
\begin{equation*}
  \inner{\psi_{k},P\Gamma\psi_{l}}_{\H}=
  \inner{\psi_{l},\left(\1_{\H}-P\right)\Gamma\psi_{k}}_{\H},\qquad k,l\in \{1,\ldots ,2N\}.
\end{equation*}
Explicitly, the matrix is given by
{\small
  \begin{align*}
    \!\!\!\!
    \!\!\!\!
    \!\!\!\!
    \!\!\!\!
    \sfM =
      \left(
      \begin{array}{cccccc}
        0              & \Mcom{1}{2}_{\H}        & \cdots & \Mcom{1}{2N-1}_{\H}  & \Mcom{1}{2N}_{\H} \\
        -\Mcom{1}{2}_{\H}   & 0                  & \cdots & \Mcom{2}{2N-1}_{\H}  & \Mcom{2}{2N}_{\H} \\
        \vdots         & \vdots                  & \ddots & \vdots         & \vdots      \\
        -\Mcom{1}{2N-1}_{\H} & -\Mcom{2}{2N-1}_{\H}  & \cdots & 0              & \Mcom{2N-1}{2N}_{\H} \\
        -\Mcom{1}{2N}_{\H}   & -\Mcom{2}{2N}_{\H}      & \cdots & -\Mcom{2N-1}{2N}_{\H} & 0 
      \end{array}
      \right).
  \end{align*}
}\par

One notices, for example, that, whenever $\psi_{2N}$ belongs to $\h\dwnP$,
the Pfaffian under consideration is zero. It is likewise zero whenever $\psi\dwnum$
belongs to $\h\dwn{\P\upperp}$. This implies that, whenever $\B(\psi_{i})$ is present in $A$,
a non-zero result requires $\B(\psi_{i})^{*}$ present to its right, and reciprocally.
Moreover, Definition \ref{Pfaffian} for the Pfaffian relies on sums of products of $N$ factors.
Therefore, there must be at least $N$ non-zero entries in $\sfM$. This only happens if,
apart from anticommutation, $A$ is of the form
$\prod\limits_{i=1}^{N}\B\left(\psi_{i}\right)\B\left(\psi_{i}\right)^{*}$,
with each $\psi_{i}$ an element of the chosen basis for $\h_{P}$.

All of the considerations of the previous paragraph allow us to
conclude, for any basis projection $P\in\fp$ with associated
quasi--free state $\omega_{P}\in\states$, that the ideal $\frakI_{\omega_{P}}$
of the GNS construction%
\footnote{We use here Bratteli--Robinson notation. See \cite[Pages 54--56]{BratteliRobinsonI}.}
associated with $\omega_{P}$, to wit,
\begin{align*}
  \frakI_{\omega_{P}} \doteq\conj{A \colon A\in\sCAR,\,\omega_{P}(A^{*}A) = 0}
\end{align*}
is the set of all elements of the $\sCAR$ which are \emph{not} of the form%
\footnote{We disregard counting elements like $\B(\varphi\dwnum)\B(\varphi\dwnum)^{*}\B(\varphi\dwndois)^{*}$, etc, with a $\B(\varphi\dwni)$ element to left of a $\B(\varphi\dwni)^{*}$, and which are not in $\frakI\dwnomega$, since, for instance, 
\begin{align*}
  (\B(\varphi\dwnum)\B(\varphi\dwnum)^{*}\B(\varphi\dwndois)^{*})^{*}
  \B(\varphi\dwnum)\B(\varphi\dwnum)^{*}\B(\varphi\dwndois)^{*} =
  \B(\varphi\dwndois)\B(\varphi\dwnum)\B(\varphi\dwnum)^{*}
  \B(\varphi\dwnum)\B(\varphi\dwnum)^{*}\B(\varphi\dwndois)^{*},
\end{align*}
and, because of the $\CAR$,
\begin{align*}
  \B(\varphi\dwnum)\B(\varphi\dwnum)^{*}\B(\varphi\dwnum)\B(\varphi\dwnum)^{*} =
  \B(\varphi\dwnum)\B(\varphi\dwnum)^{*}.
\end{align*}
}
\begin{align*}
  \B(\psi\dwnum)^{*}\cdots\B(\psi_{N})^{*}, \quad \psi\dwnum,\,\dots,\,\psi_{N} \in \h\dwnP
  ,
\end{align*}
for any $N\in\N$. For all $A\in\sCAR$, we can then construct its GNS class by
\begin{align*}
  \Psi\dwn{A} \doteq A + \frakI_{\omega_{P}},
\end{align*}
so that $\H_{\omega_{P}}$ is the completion of the vector space
$\conj{\Psi\dwn{A} \colon A\in\sCAR}$, seen as a pre--Hilbert space with inner product
given by $\inner{A,B}_{\omega_{P}}\doteq \omega_{P}(A^{*}B)$, for any $A,B\in\sCAR$.\par
The above explicit construction allows us to therefore show that $\pi_{\omega_{P}}$ and $\pi\dwnP$
are unitarily equivalent. In fact, consider the fermionic Fock space and its associated vacuum vector,
given by \eqref{eq:ferm_fock} and discussed in comments around it. For all $N\in\N$,
and all elements $\psi_{1},\ldots,\psi_{N}$ of $\h\dwnP$, consider the function given by
\begin{align*}
  \Psi\dwn{\B(\psi\dwnum)^{*}\cdots\B(\psi\dwn{N})^{*}}
  \mapsto
  \psi\dwnum\land\dots\wedge\psi\dwn{N} , \quad
  \Psi\dwn{\frakuno} (\doteq \Omega_{\omega_{P}})
  \mapsto
  \Omega,
\end{align*}
which is extended by linearity for all $\Psi_{A}$, with $A\in\sCAR$. It is clear that this
function is bounded and defined over a dense subset of $\H_{\omega_{P}}$. It therefore has a
bounded extension $U\in\BL(\H_{\omega_{P}};\calF(\h_{P}))$. Note that $U$ is unitary
and that, for all $A\in\sCAR$,
\begin{align}\label{eq:equivomegaP}
  \pi_{\omega_{P}}(A) = U^{*}\pi\dwnP(A)U,
\end{align}
whence the equivalence. This equivalence allows us to go a bit further,
and establish two other equivalences. When we consider the spaces
\begin{align*}
  \pi_{\omega_{P}}(\sCAR\uppm)\Omega_{\omega_{P}}
  ,
\end{align*}
we note that%
\footnote{
\label{footnote:elementsevenodd}
One sees that zero belongs to $\pi_{\omega_{P}}(\sCAR\upmais)\Omega_{\omega_{P}}$, and that its non--zero vectors
are of the form $\Psi\dwn{\B(\psi\dwnum)^{*}\cdots\B(\psi\dwn{N})^{*}}$, with $N$ even, or $\psi\dwn{\frakuno}$. On the other hand, one sees that zero belongs to $\pi\dwnP(\sCAR\upmais)\Omega\dwnP$, and that its non-zero vectors are of the form
$\psi\dwnum\land\dots\wedge\psi\dwn{N}$, with $N$ even, or $\Omega$.
Analogous reasoning goes for the spaces with the ``$-$'' sign. 
}
\begin{align*}
  U\left(\pi_{\omega_{P}}(\sCAR\uppm)\Omega_{\omega_{P}}\right) =
  \pi\dwnP(\sCAR\uppm)\Omega,
\end{align*}
from what we may assert that 
\begin{align*}
  U(\H_{\omega_{P}}^{\pm}) = \H\dwnP\uppm,
\end{align*}
where $\H_{\omega_{P}}\upmais$ is the closure of the space $\pi_{\omega_{P}}(\sCAR\upmais)\Omega_{\omega_{P}}$,
and $\H_{\omega_{P}}\upmenonos$ is the closure of the space $\pi_{\omega_{P}}(\sCAR\upmenonos)\Omega_{\omega_{P}}$.
Moreover, it is of notice that $\H\upmais\dwnP$ and $\H\upmenonos\dwnP$ consist,
respectively, of the vector subsets of even and odd elements of the Fock space,
which share only in common the null vector. Therefore, the fermionic Fock space splits
into the direct sum: $\calF(\h_{P})=\H\dwnP\upmais\oplus\H\dwnP\upmenonos$.
Similary, we can write $\H_{\omega_{P}}=\H_{\omega_{P}}^{+}\oplus\H_{\omega_{P}}^{-}$, in such a way that
\begin{align*}
  U = U_{+}\oplus U_{-},
\end{align*}
with $U_{\pm}\in\BL(\H_{\omega_{P}}^{\pm};\H_{P}^{\pm})$ being a unitary operator.
This allows to conclude that, if $\pi_{\omega_{P}}$
is the restriction of $\pi_{\omega_{P}}(\sCAR\upmais)$ to $\H_{\omega_{P}}^{+}$, and $\pi_{\omega_{P}}\upmenonos$
is the restriction $\pi_{\omega_{P}}(\sCAR\upmais)$ to $\H_{\omega_{P}}^{-}$, then%
\footnote{
Following footnote \ref{footnote:elementsevenodd}, notice that the even elements
from $\pi_{\omega_{P}}(\sCAR\upmais)$ (or $\pi_{P}(\sCAR\upmais)$) leave $\H_{\omega_{P}}^{\pm}$
(respectively $\H\dwnP\uppm$) invariant, for even elements of $\sCAR$ do not alter
parity.
}
\begin{align*}
  \pi_{\omega_{P}}^{+} = U_{+}^{*}\pi_{P}^{+}U_{+}, \quad
  \pi_{\omega_{P}}^{-} = U_{-}^{*}\pi_{P}^{-}U_{-},
\end{align*}
from what we conclude that $\pi_{P}^{\pm}$ and $\pi_{\omega_{P}}^{\pm}$ are
unitarily equivalent. With all the foregoing considerations, the following Theorem is easily proven:
\begin{theorem}\label{theorem:main}
Let $P_1,P_2\in\fp$ be two basis projections, and let $\omega_{P_{1}},\omega_{P_{2}}\in\states$ be the quasi--free states associated with $P_1$ and $P_2$, respectively. Let $\pi_{\omega_{P_{1}}}$ and $\pi_{\omega_{P_{2}}}$ be the GNS representations on the $\sCAR$ algebra associated with the states $\omega_{P_{1}}$ and $\omega_{P_{2}}$. Then:
\begin{enumerate}
\item\label{theorem:1} \emph{Shale--Stinespring}: $\pi_{\omega_{P_{1}}}$ and $\pi_{\omega_{P_{2}}}$ are unitarily equivalent if, and only if, $\left\Vert\omega_{P_{1}}-\omega_{P_{2}}\right\Vert_{\frakF}<\infty$;
\item\label{theorem:2} For $i\in\{1,2\}$, let $\pi\dwn{\omega\dwn{P\dwni}}\uppm$ be representations defined above and consider the \emph{$\Z_{2}$--state index}, $\sigma(\omega_{P_{1}},\omega_{P_{2}})$, given in Lemma \ref{lemma:states}. Then:
\begin{enumerate}
\item The representations $\pi\uppm\dwn{\omega\dwn{P\dwnum}}$ and $\pi\uppm\dwn{\omega\dwn{P\dwndois}}$ are irreducible;
\item The representations $\pi_{\omega_{P_{1}}}^{\pm}$ and $\pi_{\omega_{P_{2}}}^{\pm}$ are unitarily equivalent if, and only if, $\sigma(\omega_{P_{1}},\omega_{P_{2}}) = 1$ and $\left\Vert\omega_{P_{1}}-\omega_{P_{2}}\right\Vert_{\frakF}<\infty$;
\item The representations $\pi_{\omega_{P_{1}}}^{\pm}$ and $\pi_{\omega_{P_{2}}}^{\mp}$ are unitarily equivalent if, and only if, $\sigma(\omega_{P_{1}},\omega_{P_{2}}) = -1$ and $\left\Vert\omega_{P_{1}}-\omega_{P_{2}}\right\Vert_{\frakF}<\infty$.
\end{enumerate}
\end{enumerate}
\end{theorem}
\begin{proof}
\ref{theorem:1}. For $i\in\{1,2\}$, $U_{i}\in\BL(\H_{\omega_{P_{i}}};\calF(\h_{P_{i}}))$ satisfying relation \eqref{eq:equivomegaP}, suppose $\pi\dwnomegaPum$ and $\pi\dwnomegaPdois$ are unitarily equivalent, namely, there is $U\in\BL(\H_{\omega_{P_{1}}};\H_{\omega_{P_{2}}})$ such that for any $A$ in $\sCAR$,
\begin{align}\label{eq:omegaequiv}
\pi\dwnomegaPum(A) = U^{*} \pi\dwnomegaPdois(A) U.
\end{align}
Let $\tilde{U}\doteq U_{2}UU_{1}^{*}\in\BL(\calF(\h_{P_{1}});\calF(\h_{P_{2}}))$. Then, for all $A\in\sCAR$,
\begin{align*}
\tilde{U}^{*}\pi_{P_{2}}(A)\tilde{U} & = U_{1}U^{*}U_{2}^{*}\pi\dwnPdois(A)U_{2}UU_{1}^{*}\\
& = U_{1}U^{*}\pi\dwnomegaPdois(A)UU_{1}^{*}\\
& = U_{1}\pi_{\omega_{P_{1}}}(A)U_{1}^{*}\\
& = \pi_{P_{1}}(A),
\end{align*}
so that $\pi\dwnPum$ and $\pi\dwnPdois$ are equivalent. It follows by \cite[Theo. 6.14]{A87} that $P\dwnum-P\dwndois\in\BL(\H)$ is a Hilbert--Schmidt operator, and, from Lemma \ref{lemma:number}, it follows that $\norma{\omega\dwnPum-\omega\dwnPdois}\dwnfrakF<\infty$.

For the converse, if $\norma{\omega\dwnPum-\omega\dwnPdois}\dwnfrakF<\infty$, then,
by Lemma \ref{lemma:number}, $P\dwnum-P\dwndois\in\BL(\H)$ is Hilbert--Schmidt, and,
by \cite[Theo. 6.14]{A87}, $\pi\dwnPum$ and $\pi\dwnPdois$ are unitarily equivalent. Let $\tilde{U}$ be the unitary which implements this equivalence, that is, for all $A\in\sCAR$,
\begin{align*}
\pi\dwnPum(A) = \tilde{U}^{*} \pi\dwnPdois(A) \tilde{U}.
\end{align*}
Then, for $U=U_{2}^{*}\tilde{U}U_{1}$, calculations similar to the foregone show that $\tilde{U}$ establishes the equivalence between $\pi\dwnomegaPum$ and $\pi\dwnomegaPdois$, like that given by \eqref{eq:omegaequiv}. Item \ref{theorem:1} is thus proven.

\ref{theorem:2} (a) By the above discussion, for $i\in\{1,2\}$, $\pi\dwn{\omega{\dwn{P\dwni}}}\uppm$, is unitarily equivalent to $\pi\dwn{P\dwni}\uppm$, and each of the latter is, by \cite[Theor. 6.15]{A87}, an irreducible representation. \ref{theorem:2} (b) Follows from an argument similar to that used for Item \ref{theorem:1}. In case $\pi\dwnomegaPum\upmais$ and $\pi\dwnomegaPdois\upmais$ are unitarily equivalent, we can show, by an appropriate choice of unitary operator, that $\pi\dwnPum\upmais$ and $\pi\dwnPdois\upmais$ are unitarily equivalent. This uses the fact that $\pi\dwnomegaPum\upmais$ and $\pi\dwnPum\upmais$ are unitarily equivalent, and so are $\pi\dwnomegaPdois\upmais$ and $\pi\dwnPdois\upmais$. From this, it follows, by \cite[Theor. 6.15 (2)]{A87}, that $P\dwnum-P\dwndois$ is Hilbert--Schmidt, and that $\sigma(P\dwnum,P\dwndois)=1$. By Lemma \ref{lemma:number}, this implies that $\norma{\omega\dwnPum-\omega\dwnPdois}\dwnfrakF<\infty$, and that $\sigma(\omega\dwnPum,\omega\dwnPdois)=1$. For the converse, if $\norma{\omega\dwnPum-\omega\dwnPdois}\dwnfrakF<\infty$ and $\sigma(\omega\dwnPum,\omega\dwnPdois)=1$, Lemma \ref{lemma:number} and \cite[Theo. 6.15 (2)]{A87} allow us to conclude that $\pi\dwnPum\upmais$ and $\pi\dwnPdois\upmais$ are equivalent, from what follows, by an appropriate choice of unitary transformation, that $\pi\dwnomegaPum\upmais$ and $\pi\dwnomegaPdois\upmais$ are unitarily equivalent. \ref{theorem:2} (c) is proven in the exact same way as \ref{theorem:2} (b).
\end{proof}

\bibliography{books}
\bibliographystyle{amsalpha}

\end{document}